\UseRawInputEncoding
\documentclass{article}
\pdfoutput=1
\usepackage{graphicx}
\usepackage{times}
\usepackage{booktabs}
\usepackage{mathptmx}
\usepackage{amsmath,amssymb,amsfonts}%
\usepackage{amsthm}%
\usepackage{mathrsfs}%
\usepackage{algorithm}%
\usepackage{algorithmicx}%
\usepackage{algpseudocode}%
\usepackage{color}
\usepackage{xcolor}
\newtheorem{theorem}{Theorem}
\newtheorem{assumption}{Assumption}[section]
\newtheorem{lemma}{Lemma}[section] 

\title{Theoretical Analysis of Impact of Delayed Updates on Decentralized Federated Learning}

\author{Yong Zeng\footnote{Sichuan Tengden Technology Co.,Ltd. E-mail: run8686@aliyun.com}, Siyuan Liu\footnote{College of Information Engineering, Inner Mongolia University of Technology, Hohhot 100080, China. Email: siyuan\_liu2022@foxmail.com}, Zhiwei Xu\footnote{Corresponding Author.
\\ Institute of Computing Technology Chinese Academy of Sciences, \mbox{Beijing 53035, China}. \\ Haihe Laboratory of Information Technology Application Innovation, Tianjin 300350, China. Email: xuzhiwei2001@ict.ac.cn} and Jie Tian\footnote{Department of Computer Science, New Jersey Institute of Technology, Newark NJ 07102, USA. Email: jt66@njit.edu}}

\begin{document}
\date{}
\maketitle

\begin{abstract}
 Decentralized Federated learning is a distributed edge intelligence framework by exchanging parameter updates instead of training data among participators, in order to retrain or fine-tune deep learning models for mobile intelligent applications. Considering the various topologies of edge networks in mobile internet, the impact of transmission delay of updates during model training is non-negligible for data-intensive intelligent applications on mobile devices, e.g., intelligent medical services, automated driving vehicles, etc..  
To address this problem, we analyze the impact of delayed updates for decentralized federated learning, and provide a theoretical bound for these updates to achieve model convergence. 
Within the theoretical bound of updating period, the latest versions for the delayed updates are reused to continue aggregation, in case the model parameters from a specific neighbor are not collected or updated in time. 

\end{abstract}

\section{INTRODUCTION}\label{sec1}

Consider a non-convex stochastic optimization problems of the form: 
\begin{equation}\label{in}
\begin{array}{ll}
\min & f(x):=\sum_{k=1}^K \mathbb{E}_{\xi \sim \mathcal{D}} g_k(x; \xi) \\
\text { s.t. } & x \in X
\end{array}
\end{equation}
where $g_k(x)$ are a set of smooth, possibly non-convex functions and represent the loss of a model parameterized by $x$ on the datum $\xi$. $\mathcal{D}$ represents a finite dataset of size $n$ or a population distribution depending on the application. Eq.\ref{in} is a machine learning (ML) training objective, typically computed by stochastic gradient descent (SGD) algorithm\cite{bottou2010large} and its variants, \emph{e.g.}, momentum SGD, Adam, etc. Such stochastic optimization problems continue to grow rapidly, both in terms of the number of model parameters in ML and the quantity of data. Given the recent growth in the size of models and available training data, there is an escalating need to employ parallel optimization algorithms to handle large-scale data and harness the benefits of data distributed across different machines. In a distributed setting, also known as data-parallel training, optimization is spread across numerous computational devices operating in parallel, such as cores or GPUs on a cluster, to expedite the training process. Every worker computes gradients on a subset of the training data, and the resulting gradients are aggregated (averaged) on a server. Parallel and distributed versions of SGD are becoming increasingly important\cite{zinkevich2010parallelized}. 

There are several ways of employing parallelism to solve Eq.\ref{in}. Among them, federated learning (FL) is a privacy-preserving distributed machine learning paradigm. In FL, Each worker performs local model training using its own data and sends only model updates (gradients) back to the central server. These updates are aggregated to refine the global model, and the entire process is iterative, with multiple rounds of local training and aggregation, gradually improving the model's performance without exposing sensitive data. This is a commonly used in practice\cite{mcmahan2017communication}. In a centralized model, the central server acts as a single point of control, raising concerns about data security and monopolization of learning processes. The central entity is compromised or misused will result in privacy breaches or a concentration of power. In contrast, decentralized federated learning (DFL) distributes the control and aggregation of model updates across multiple nodes or entities in a network. This decentralization enhances security, prevents single points of failure, and fosters collaboration while preserving data privacy. However, since traditional FL method waits for all machines to finish computing their gradient estimates before updating, it proceeds only at the speed of the slowest machine.

When a large volume of data is distributed across computational nodes, local computations can be costly and time-consuming. If synchronous algorithms are used, the slowest node can significantly impede the overall system's performance. Several potential sources of delay exist: nodes may possess heterogeneous hardware with varying computational throughputs\cite{wu2017stochastic}, network latency can slow down gradient communication, and nodes may even drop out\cite{hung2018delay}. In many naturally parallel settings, slower "lagging" nodes can arise, including scenarios involving multiple GPUs\cite{yi2020fast} or training ML models in the cloud\cite{luo2020plink}. Sensitivity to these lagging nodes poses a significant challenge for synchronous algorithms. In fact, in fully decentralized environments where there is no clock synchronization, minimal coordination between distributed nodes, and no guaranteed mechanism for reliable communication, ideal distributed algorithms should be robust enough to handle various asynchronous sources while still producing high-quality solutions within a reasonable time frame.  

Since the seminal work of Bertsekas\cite{bertsekas2015parallel} and Tsitsiklis\cite{tsitsiklis1986distributed}, there has been a large body of literature focusing on asynchronous implementation of various distributed schemes. Hao et al.\cite{zhu2010distributed} discussed the performance differences between distributed demodulation and centralized counterparts, showing that there are only a few consensus iterations, suffice for the distributed demodulators to approach the performance of their centralized counterparts. Nedić et al.\cite{nedic2001distributed} discuss a distributed asynchronous subgradient method for minimizing convex functions consisting of the sum of a large number of component functions. At each step, some outdated gradients can be updated.
In \cite{liu2014asynchronous}, an asynchronous parallel stochastic coordinate descent algorithm is developed that achieves linear convergence speed, where the update of each block can exploit delayed gradient information.
Agarwal et al.\cite{agarwal2011distributed} show that uncertainty can also be tolerated in stochastic optimization. Furthermore, they demonstrate that the convergence speed depends on delayed stochastic gradient information.
Zhou et al. \cite{zhou2018distributed} discuss the delay problem in distributed computing, allowing the delay to grow with time, but only showing asymptotic convergence.
In typical FL applications, clients or workers often exhibit significant disparities in computational capabilities/speed. Consequently, practitioners favor the utilization of asynchronous algorithms in FL\cite{nguyen2022federated, avdiukhin2021federated, gu2021fast}, with much effort directed towards addressing the unequal participation rates among diverse clients through the implementation of variance reduction techniques on the server side\cite{yan2020distributed, yang2022anarchic}.

In this paper, we analyze the impact of delayed updates for decentralized federated learning, and provide a theoretical bound for these updates to achieve model convergence. Within the theoretical bound of updating period, the latest versions for the delayed updates are reused to continue aggregation, in case the model parameters from a specific neighbor are not collected or updated in time.

The rest of this paper is organized as follows.  The problem formulation are studied in In Section 2. In Section 3, we theoretically analyze the impact of delayed updates for decentralized federated learning delay. Finally, we give our conclusions about this work.

\section{Problem Formulation}

Federated learning is an incremental optimization algorithm in which a large population of devices collaboratively trains a neural network model. The traditional FL distributed architecture is based on a central server with several clients. Our work is built on a decentralized federal learning framework, which allows the client to update parameters based on local parameters as well as those of its neighbors. Assume that the interaction topology of the network is constructed as a directed graph $G=(C, E)$, where $C=\{1, 2, \dots , K\}$ are nodes in the graph and $E$ is a connected edge set between nodes. The neighbor set of client $k$ is denoted as $N_{\overline{k}}$, with cardinality $\lvert N_{\bar{k}}\rvert$. Notice that we include client $k$ in the $M_k$ while ${N_{\overline{k}}}= M_k\backslash \{k\}$ does not. The training dataset $\mathcal{D}={\left \{ x_i, y_i \right \} }_{i=1}^{ \lvert \mathcal{D} \rvert }$ distributed on $K$ clients, where $\lvert \mathcal{D}\rvert$ is the total number of training samples. $x_i$ is the i-th sample, $y_i$ is the corresponding label. Dataset on every client $k,k=1\dots K$ is denoted by $\left \{ \mathcal{D}_k\right \}_{k=1}^K$, where  $ {\textstyle \bigcup_{k=1}^{K}\mathcal{D}_k}=\mathcal{D}$. Then the loss function of the i-th sample can be denoted by $f_i(x_i,y_i;w)$, where $w$ is the model parameters. Function $F_k(w)$ is the loss function on client $k$. In Decentralized Federated Learning, every client trains the global model by obtaining the parameters of neighbors instead of data. We consider the following optimization problem:

\begin{equation}\label{fl obective function}
   \min \sum_{k=1}^{K}\frac{ {\textstyle \lvert \mathcal{D}_k \rvert } }{\lvert \mathcal{D} \rvert }F_k(w),\quad{\tiny }  F_k(w)=\frac{ {\textstyle \sum_{i\in\mathcal{D}_k}^{}f_i(x_i,y_i;w)} }{\lvert \mathcal{D}_k \rvert}.
\end{equation}

Among the state of the art, CFA\cite{26} provides a typical paradigm for decentralized federated learning. Each client trains its local model and sends the local model parameters to its neighbor as an update. Each client computes the weighted average of the received updates, and trains with parameters updated for the next iteration. After several rounds of model updating and aggregation, a converged global model for every client is obtained. 

For the $t_{th}$ iteration, We assume that $w_k^t$ is the parameters vector of the local model in client $k$. After receiving the parameters from neighbors, the client obtains the aggregated model
\begin{equation}
  \psi_{k}^{t}=w_{k}^{t}+\epsilon_{t} \sum_{k=1}^{K} \frac{\mathcal{D}_{k}}{\mathcal{D}} (w_{k}^{t}-w_{p\left\{p \in N_{\bar{k}}\right\}}^{t}).
\end{equation}
The local parameter vector $w_k^t$ in CFA is updated as 
\begin{equation}
  w_k^{t+1}\leftarrow \psi_k^t-\eta\nabla F_k(w_k^t),
\end{equation}
where $\eta$ is learning rate.

\section{Exploring The Impact of Delayed Updates}

In this section, we analyze the convergence of delayed updates for decentralized federated learning.  Specifically, we formulate this problem \ref{fl obective function} into a linear constraint problem as follows
\begin{equation}\label{lincon}
    \begin{aligned}
        \min &\sum_{k=1}^{K}\frac{\lvert\mathcal{D}_{k}\rvert }{\lvert \mathcal{D}\rvert}F_{k}(w_k) \\
    \text { s.t. } &  w_{k}=w, \forall k=1, \cdots, K, && w \in \mathcal{W}. 
    \end{aligned}
\end{equation}
Here, $F_k(\cdot)$’s  are a set of loss functions defined in Eq.\ref{fl obective function}, $w$ is a global parameter for aggregating the parameters of neighbor nodes, and $w_k$ is a parameter of client $k$. We define a dual variable of $w$ 
that distribute on clients and update as follows
\begin{equation}\label{dv}
    \lambda_{k}^{t+1}=\lambda_{k}^{t}+\eta_{k}\left(w_{k}^{t+1}-w^{t+1}\right), \forall k \text {. }
\end{equation}

Then, transform the linear constraint problem into an unconstrained augmented Lagrange function
\begin{equation}\label{alf}
    L\left(\left\{w_{k}\right\}, w ; \lambda\right)=  \sum_{k=1}^{K}\frac{\lvert\mathcal{D}_{k}\rvert }{\lvert \mathcal{D}\rvert}F_{k}(w_k)+\sum_{k=1}^{K}\left\langle \lambda_{k}, w_{k}-w\right\rangle
 +\sum_{k=1}^{K} \frac{\eta_{k}}{2}\left\|w_{k}-w\right\|^{2},
\end{equation}
where $\eta_{k}>0$ is some constant. 

To analyze the convergence of Eq.\ref{alf}, we make the following assumptions. First, we assume that the gradient of loss function $\nabla g_k(\cdot)$ for each client $k$ conforms to the Lipschitz continuity.
\begin{assumption}\label{assump1}
Lipschitz continuity. There exists a positive constant $M_k$ such that for any client $k$, $x,y\in \mathcal{W}$, $\left\|\nabla F_{k}(x)-\nabla F_{k}(y)\right\| \leq M_k\|x-y\|$.
\end{assumption}
Next, we make an assumption about the upper bound on the latency of asynchronous communication between clients.
\begin{assumption}\label{assump2} 
There exists a finite constant $T_k$ such that $k$, $t-t(k)\leq T_k$, where $t(k)$ is the index of parameter from neighbors which is used by the client $k$.
\end{assumption}

According to the above two assumptions, we have results as follows.

\begin{lemma}\label{lemma1}
The update of the dual variable $\lambda_k$ satisfies
\begin{equation}
    \left\|\lambda_{k}^{t+1} - \lambda_{k}^{t}\right\|^{2} \leq M_{k}^{2}\left(T_{k}+1\right) \sum_{m=0}^{T_{k}}\left\|w^{t+1-m}-w^{t-m}\right\|^{2},
\end{equation}
where $M_k$ is Lipschitz constant in Assumption \ref{assump1}, $T_k$ is defined in Assumption \ref{assump2} and $\lambda_k$.
\end{lemma}
\begin{proof}
In the iterative process of the client, the following is true
    \begin{equation}
        \nabla F_{k}\left(w^{[t+1](k)}\right)+\lambda_{k}^{t}+\eta_{k}\left(w_{k}^{t+1}-w^{t+1}\right)=0
    \end{equation}
Equivalently, combined with Eq.\ref{dv} there is
\begin{equation}\label{eq12}
    \nabla F_{k}\left(w^{[t](k)}\right)=-\lambda_{k}^{t}
\end{equation}
If the latest parameter $w_k$ of client $k$ are not updated in the next iteration $t+1$, we can obtain
\begin{equation}
    \left\|w_{k}^{t+1}-w_{k}^{t}\right\|^{2}=0
\end{equation}
Conversely, if the latest parameter $w_k$ of client $k$ are updated in the next iteration $t+1$, we can obtain
\begin{equation}
    \begin{array}{l}
\left\{\begin{array}{ll}
\lvert [t+1](k)-[t](k)\rvert\leq t+1-[t](k) \leq T_{k}+1, & \text { if }[t+1](k) \geq[t](k), \\
\lvert[t+1](k)-[t](k)\rvert\leq t+1-[t+1](k) \leq T_{k}, & \text { otherwise }.
\end{array}\right.
\end{array}
\end{equation}
Considering the above two cases, we have
\begin{equation}\label{l1}
    \begin{aligned}
\left\|\lambda_{k}^{t+1}-\lambda_{k}^{t}\right\| & =\left\|\nabla F_{k}\left(w^{[t+1](k)}\right)-\nabla F_{k}\left(w^{[t](k)}\right)\right\| \\
& \leq M_{k} \sum_{m=0}^{T_{k}}\left\|w^{t+1-m}-w^{t-m}\right\|
\end{aligned}
\end{equation}
Eq.\ref{l1} implies that the result is obtained.
\end{proof}

Next, in order to upper bound the augmented Lagrangian, we define auxiliary functions as follows
\begin{equation}
\begin{array}{l}
p_{k}\left(w_{k} ; w^{t+1}, \lambda^{t}\right)=F_{k}\left(w_{k}\right)+\left\langle \lambda_{k}^{t}, w_{k}-w^{t+1}\right\rangle+\frac{\eta_{k}}{2}\left\|w_{k}-w^{t+1}\right\|^{2} \\
q_{k}\left(w_{k} ; w^{t+1}, \lambda^{t}\right)=F_{k}\left(w^{t+1}\right)+\left\langle\nabla F_{k}\left(w^{t+1}\right), w_{k}-w^{t+1}\right\rangle \\ 
\quad \quad \quad \quad \quad \quad \quad \quad + \left\langle \lambda_{k}^{t}, w_{k}-w^{t+1}\right\rangle+\frac{\eta_{k}}{2}\left\|w_{k}-w^{t+1}\right\|^{2} \\
\bar{q}_{k}\left(w_{k} ; w^{t+1}, \lambda^{t}\right)=F_{k}\left(w^{t+1}\right)+\left\langle\nabla F_{k}\left(w^{[t+1](k)}\right), w_{k}-w^{t+1}\right\rangle \\ 
\quad \quad \quad \quad \quad \quad  \quad \quad+ \left\langle \lambda_{k}^{t}, w_{k}-w^{t+1}\right\rangle+\frac{\eta_{k}}{2}\left\|w_{k}-w^{t+1}\right\|^{2} .
\end{array}
\end{equation}

According to assumptions and Lemma \ref{lemma1}, we can obtain the properties of the function $p_k(\cdot)$ as follows.
\begin{lemma}\label{lemma2}
\begin{equation}
  \begin{array}{l}
p_{k}\left(w_{k}^{t+1} ; w^{t+1}, \lambda^{t}\right)-p_{k}\left(w_{k}^{t} ; w^{t+1}, \lambda^{t}\right) \\
\leq-\left(\frac{\eta_{k}}{2}-\frac{7}{2} M_{k}\right)\left\|w_{k}^{t}-w_{k}^{t+1}\right\|^{2}+\frac{M_{k} T_{k}}{2} \sum_{i=0}^{T_{k}-1}\left\|w^{t+1-i}-w^{t-i}\right\|^{2} \\
\quad+\frac{7 M_{k}}{2 \eta_{k}^{2}}\left\|\lambda_{k}^{t+1}-\lambda_{k}^{t}\right\|^{2}, \quad k=1, \cdots, K .
\end{array}  
\end{equation}
\end{lemma}
\begin{proof}
    From function $p_k(\cdot)$ and $q_k(\cdot)$ we defined above, we have
    \begin{equation}
        p_{k}\left(w_{k} ; w^{t+1}, \lambda^{t}\right) \leq q_{k}\left(w_{k} ; w^{t+1}, \lambda^{t}\right)+\frac{M_{k}}{2}\left\|w_{k}-w^{t+1}\right\|^{2}, \forall w_{k}, k=1, \cdots, K.
    \end{equation}
    Using function $\bar{q}_{k}(\cdot)$ we defined, we can obtain
    \begin{equation}
        w_{k}^{t+1}=\arg \min _{w_{k}} \bar{q}_{k}\left(w_{k} ; w^{t+1}, \lambda^{t}\right).
    \end{equation}
Due to the strong convexity of $\bar{q}_{k}\left(w_{k} ; w^{t+1}, \lambda^{t}\right)$ with respect to $w_k$, we have
\begin{equation}\label{eq3}
    \begin{array}{l}
\bar{q}_{k}\left(w_{k}^{t+1} ; w^{t+1}, \lambda^{t}\right)-\bar{q}_{k}\left(w_{k}^{t} ; w^{t+1}, \lambda^{t}\right) \\
\leq-\frac{\eta_{k}}{2}\left\|w_{k}^{t}-w_{k}^{t+1}\right\|^{2}, \forall k,
\nabla \bar{q}_{k}\left(w_{k}^{t+1} ; w^{t+1}, \lambda^{t}\right)=0.
\end{array}
\end{equation}
Due to the strong convexity of $q_{k}\left(w_{k} ; w^{t+1}, \lambda^{t}\right)$, we also have
\begin{equation}\label{eq2}
    \begin{array}{l}
q_{k}\left(w_{k}^{t+1} ; w^{t+1}, \lambda^{t}\right) \\ \leq  q_{k}\left(w_{k}^{t} ; w^{t+1}, \lambda^{t}\right)+\left\langle\nabla q_{k}\left(w_{k}^{t+1} ; w^{t+1}, \lambda^{t}\right), w_{k}^{t+1}-w_{k}^{t}\right\rangle-\frac{\eta_{k}}{2}\left\|w_{k}^{t+1}-w_{k}^{t}\right\|^{2} \\
\leq  q_{k}\left(w_{k}^{t} ; w^{t+1}, \lambda^{t}\right)+M_{k}\left\|w^{[t+1](k)}-w^{t+1}\right\|\left\|w_{k}^{t+1}-w_{k}^{t}\right\|-\frac{\eta_{k}}{2}\left\|w_{k}^{t+1}-w_{k}^{t}\right\|^{2} \\
\leq  q_{k}\left(w_{k}^{t} ; w^{t+1}, \lambda^{t}\right)+\frac{M_{k} T_{k}}{2} \sum_{i=0}^{T_{k}-1}\left\|w^{t+1-i}-w^{t-i}\right\|^{2}-\frac{\eta_{k}-M_{k}}{2}\left\|w_{k}^{t+1}-w_{k}^{t}\right\|^{2}.
\end{array}
\end{equation}
In addition, there are series of inequalities about function $p_k(\cdot)$ and $q_k(\cdot)$
\begin{equation}\label{eq1}
    \begin{array}{l}
q_{k}\left(w_{k}^{t} ; w^{t+1}, \lambda^{t}\right)-p_{k}\left(w_{k}^{t} ; w^{t+1}, \lambda^{t}\right) \\
=F_{k}\left(w^{t+1}\right)+\left\langle\nabla F_{k}\left(w^{t+1}\right), w_{k}^{t}-w^{t+1}\right\rangle 
+\left\langle \lambda_{k}^{t}, w_{k}^{t}-w^{t+1}\right\rangle+\frac{\eta_{k}}{2}\left\|w_{k}^{t}-w^{t+1}\right\|^{2} \\
\quad \quad-\left(F_{k}\left(w_{k}^{t}\right)+\left\langle \lambda_{k}^{t}, w_{k}^{t}-w^{t+1}\right\rangle+\frac{\eta_{k}}{2}\left\|w_{k}^{t}-w^{t+1}\right\|^{2}\right) \\
\leq\left\langle\nabla F_{k}\left(w^{t+1}\right)-\nabla F_{k}\left(w_{k}^{t}\right), w_{k}^{t}-w^{t+1}\right\rangle+\frac{M_{k}}{2}\left\|w_{k}^{t}-w^{t+1}\right\|^{2} \\
\leq 3 M_{k}\left(\left\|w_{k}^{t}-w_{k}^{t+1}\right\|^{2}+\left\|w_{k}^{t+1}-w^{t+1}\right\|^{2}\right).
\end{array}
\end{equation}
Combined Eq.\ref{eq3}, \ref{eq2}, \ref{eq1} and Assumption \ref{assump1}, we can obtain

\begin{equation}
    \begin{aligned}
& p_k\left(w_k^{t+1} ; w^{t+1}, \lambda^t\right)-p_k\left(w_k^t ; w^{t+1}, \lambda^t\right) \\
& \leq q_k\left(w_k^{t+1} ; w^{t+1}, \lambda^t\right)-q_k\left(w_k^t ; w^{t+1}, \lambda^t\right)+\frac{M_k}{2}\left\|w^{t+1}-w_k^{t+1}\right\|^2 \\
& \quad+q_k\left(w_k^t ; w^{t+1}, \lambda^t\right)-p_k\left(w_k^t ; w^{t+1}, \lambda^t\right) \\
& \leq q_k\left(w_k^{t+1} ; w^{t+1}, \lambda^t\right)-q_k\left(w_k^t ; w^{t+1}, \lambda^t\right)+\frac{M_k}{2}\left\|w^{t+1}-w_k^{t+1}\right\|^2 \\
& + 3 M_{k}\left(\left\|w_{k}^{t}-w_{k}^{t+1}\right\|^{2}+\left\|w_{k}^{t+1}-w^{t+1}\right\|^{2}\right)\\
& \leq q_k\left(w_k^{t+1} ; w^{t+1}, \lambda^t\right)-q_k\left(w_k^t ; w^{t+1}, \lambda^t\right)+\frac{7M_k}{2}\left\|w^{t+1}-w_k^{t+1}\right\|^2 \\
& + 3 M_{k}\left\|w_{k}^{t}-w_{k}^{t+1}\right\|^{2}\\
& \leq \frac{M_{k} T_{k}}{2} \sum_{i=0}^{T_{k}-1}\left\|w^{t+1-i}-w^{t-i}\right\|^{2}-\frac{\eta_{k}-M_{k}}{2}\left\|w_{k}^{t+1}-w_{k}^{t}\right\|^{2}\\
& +\frac{7M_k}{2}\left\|w^{t+1}-w_k^{t+1}\right\|^2 + 3 M_{k}\left\|w_{k}^{t}-w_{k}^{t+1}\right\|^{2}\\
& \leq-\frac{\eta_k-M_k}{2}\left\|w_k^t-w_k^{t+1}\right\|^2+\frac{M_k T_k}{2} \sum_{i=0}^{T_k-1}\left\|w^{t+1-i}-w^{t-i}\right\|^2 \\
& \quad+\frac{7 M_k}{2 \eta_k^2}\left\|\lambda_k^{t+1}-\lambda_k^t\right\|^2+3 M_k\left\|w_k^t-w_k^{t+1}\right\|^2 .
\end{aligned}
\end{equation}

The desired result then follows.

\end{proof}

Using function $p_k(\cdot)$ we defined above, we have
\begin{equation}
    L\left(\left\{w_{k}^{t+1}\right\}, w^{t+1} ; \lambda^{t}\right)=\sum_{k=1}^{K} p_{k}\left(w_{k}^{t+1} ; w^{t+1}, \lambda^{t}\right).
\end{equation}
In the same setting as Lemma \ref{lemma2}, we have
\begin{lemma}\label{lemma3}
The augmented Lagrangian function satisfies the following properties in each round of iterations
    \begin{equation}
        \begin{array}{l}
L\left(\left\{w_{k}^{t+1}\right\}, w^{t+1} ; \lambda^{t+1}\right)-L\left(\left\{w_{k}^{1}\right\}, w^{1} ; \lambda^{1}\right) \\
\leq-\sum_{i=1}^{t} \sum_{k=1}^{K}\left(\frac{\eta_{k}-7 M_{k}}{2}\right)\left\|w_{k}^{i+1}-w_{k}^{i}\right\|^{2} \\
-\sum_{i=1}^{t} \sum_{k=1}^{K} \left(\eta_{k}-2\left(\frac{1}{\eta_{k}}-\frac{7 M_{k}}{2 \eta_{k}^{2}}\right) M_{k}^{2}\left(T_{k}+1\right)^{2}+M_{k} T_{k}^{2}\right)\left\|w^{i+1}-w^{i}\right\|^{2}.
\end{array}
    \end{equation}
\end{lemma}
\begin{proof}
    
    First, we restrict the continuous difference $L\left(\left\{w_{k}^{t+1}\right\}, w^{t+1} ; \lambda^{t+1}\right)-L\left(\left\{w_{k}^{1}\right\}, w^{1} ; \lambda^{1}\right)$.
We first separate the difference into
\begin{equation}\label{eq25}
    \begin{aligned}
& L\left(\left\{w_k^{t+1}\right\}, w^{t+1} ; \lambda^{t+1}\right)-L\left(\left\{w_k^t\right\}, w^t ; \lambda^t\right) \\
& =\left(L\left(\left\{w_k^{t+1}\right\}, w^{t+1} ; \lambda^{t+1}\right)-L\left(\left\{w_k^{t+1}\right\}, w^{t+1} ; \lambda^t\right)\right) \\
& \quad+\left(L\left(\left\{w_k^{t+1}\right\}, w^{t+1} ; \lambda^t\right)-L\left(\left\{w_k^t\right\}, w^t ; \lambda^t\right)\right) .
\end{aligned}
\end{equation}
    
The first term is expressed in Eq.\ref{eq25} as

    \begin{equation}
        \begin{aligned}
& L\left(\left\{w_k^{t+1}\right\}, w^{t+1} ; \lambda^{t+1}\right)-L\left(\left\{w_k^{t+1}\right\}, w^{t+1} ; \lambda^t\right) \\
& =\sum_{k=1}^K \frac{1}{\eta_k}\left\|\lambda_k^{t+1}-\lambda_k^t\right\|^2 .
\end{aligned}
    \end{equation}

To bound the second term in Eq.\ref{eq25}, 

\begin{equation}\label{eq27}
\begin{aligned}
& L\left(\left\{w_k^{t+1}\right\}, w^{t+1} ; \lambda^t\right)-L\left(\left\{w_k^t\right\}, w^t ; \lambda^t\right) \\
& =L\left(\left\{w_k^{t+1}\right\}, w^{t+1} ; \lambda^t\right)-L\left(\left\{w_k^t\right\}, w^{t+1} ; \lambda^t\right)+L\left(\left\{w_k^t\right\}, w^{t+1} ; \lambda^t\right)-L\left(\left\{w_k^t\right\}, w^t ; \lambda^t\right) \\
& =\sum_{k=1}^K\left(p_k\left(w_k^{t+1} ; w^{t+1}, \lambda^t\right)-p_k\left(w_k^t ; w^{t+1}, \lambda^t\right)\right)+L\left(\left\{w_k^t\right\}, w^{t+1} ; \lambda^t\right)-L\left(\left\{w_k^t\right\}, w^t ; \lambda^t\right) \\
& \leq-\sum_{k=1}^K\left[\left(\frac{\eta_k}{2}-\frac{7}{2} M_k\right)\left\|w_k^t-w_k^{t+1}\right\|^2-\frac{M_k T_k}{2} \sum_{i=0}^{T_k-1}\left\|w^{t+1-i}-w^{t-i}\right\|^2\right. \\
& -\frac{7 M_k}{2 \eta_k^2}\left\|\lambda_k^{t+1}-\lambda_k^t\right\|^2-\frac{1}{2} \sum_{k=1}^K \eta_k\left\|w^{t+1}-w^t\right\|^2.
\end{aligned}
\end{equation}

The last inequality in Eq.\ref{eq27} is obtained by Lemma~\ref{lemma2} and the strong convexity of $L({w_t^k}, w; \lambda_t)$ with respect to the variable $w$ at $w = w^{t+ 1}$.

\begin{equation}
    \begin{aligned}
& L\left(\left\{w_k^{t+1}\right\}, w^{t+1} ; \lambda^{t+1}\right)-L\left(\left\{w_k^t\right\}, w^t ; \lambda^t\right) \leq \sum_{k=1}^K\left[-\left(\frac{\eta_k}{2}-\frac{7}{2} M_k\right)\left\|w_k^t-w_k^{t+1}\right\|^2\right. \\
& \left.\quad+\frac{M_k T_k}{2} \sum_{i=0}^{T_k-1}\left\|w^{t+1-i}-w^{t-i}\right\|^2\right]-\frac{1}{2} \sum_{k=1}^K \eta_k\left\|w^{t+1}-w^t\right\|^2 \\
& \quad+\sum_{k=1}^K\left(\frac{1}{\eta_k}+\frac{7 M_k}{2 \eta_k^2}\right)\left(M_k^2\left(T_k+1\right) \sum_{i=0}^{T_k}\left\|w^{t+1-i}-w^{t-i}\right\|^2\right) .
\end{aligned}
\end{equation}

Next, combining the above two inequalities(Eq.\ref{eq25} and Eq.\ref{eq27}) and using Lemma \ref{lemma1}, we get the following inequality:

\begin{equation}\label{eq29}
    \begin{aligned}
& L\left(\left\{w_k^{t+1}\right\}, w^{t+1} ; \lambda^{t+1}\right)-L\left(\left\{w_k^t\right\}, w^t ; \lambda^t\right) \leq \sum_{k=1}^K\left[-\left(\frac{\eta_k}{2}-\frac{7}{2} M_k\right)\left\|w_k^t-w_k^{t+1}\right\|^2\right. \\
& \left.\quad+\frac{M_k T_k}{2} \sum_{i=0}^{T_k-1}\left\|w^{t+1-i}-w^{t-i}\right\|^2\right]-\frac{1}{2} \sum_{k=1}^K \eta_k\left\|w^{t+1}-w^t\right\|^2 \\
& \quad+\sum_{k=1}^K\left(\frac{1}{\eta_k}+\frac{7 M_k}{2 \eta_k^2}\right)\left(M_k^2\left(T_k+1\right) \sum_{i=0}^{T_k}\left\|w^{t+1-i}-w^{t-i}\right\|^2\right) .
\end{aligned}
\end{equation}
Then, for any given $t$, the difference $L({w^{t+1}_k}, w^{t+1}; \lambda ^{t+1})-L({w^1_k}, w^1; \lambda^1)$ is obtained by summing over all iterations Eq.\ref{eq29} acquired:
\begin{equation}
    \begin{aligned}
& L\left(\left\{w^{t+1}\right\}, w^{t+1} ; \lambda^{t+1}\right)-L\left(\left\{w_k^1\right\}, w^1 ; \lambda^1\right) \\
& \leq-\sum_{i=1}^t \sum_{k=1}^K\left(\frac{\eta_k}{2}-\frac{7}{2} M_k\right)\left\|w_k^{i+1}-w_k^i\right\|^2 \\
&-\sum_{i=1}^t \sum_{k=1}^K\left(\frac{\eta_k}{2}-\left(\frac{1}{\eta_k}+\frac{7 M_k}{2 \eta_k^2}\right) M_k^2\left(T_k+1\right)^2-\frac{M_k T_k^2}{2}\right)\left\|w^{i+1}-w^i\right\|^2 \\
&:=-\sum_{i=1}^t \sum_{k=1}^K \frac{\eta_k-7 M_k}{2}\left\|w_k^{i+1}-w_k^i\right\|^2 \\
&-\sum_{i=1}^t \sum_{k=1}^K \left(\eta_{k}-2\left(\frac{1}{\eta_{k}}-\frac{7 M_{k}}{2 \eta_{k}^{2}}\right) M_{k}^{2}\left(T_{k}+1\right)^{2}+M_{k} T_{k}^{2}\right)\left\|w^{i+1}-w^i\right\|^2 .
\end{aligned}
\end{equation}
\end{proof}

To simplify the following analysis, we define a variable $\alpha_k$, 
\begin{equation}\label{alpha}
    \alpha_k:=\eta_{k}-2\left(\frac{1}{\eta_{k}}-\frac{7 M_{k}}{2 \eta_{k}^{2}}\right) M_{k}^{2}\left(T_{k}+1\right)^{2}+M_{k} T_{k}^{2}.
\end{equation}

Additionally, by using Eq.\ref{alpha}, we can convert Lemma~\ref{lemma3} to

\begin{equation}\label{eq30}
            \begin{array}{l}
L\left(\left\{w_{k}^{t+1}\right\}, w^{t+1} ; \lambda^{t+1}\right)-L\left(\left\{w_{k}^{1}\right\}, w^{1} ; \lambda^{1}\right) \\
\leq-\sum_{i=1}^{t} \sum_{k=1}^{K}\left(\frac{\eta_{k}-7 M_{k}}{2}\right)\left\|w_{k}^{i+1}-w_{k}^{i}\right\|^{2}
-\sum_{i=1}^{t} \sum_{k=1}^{K} \alpha_k\left\|w^{i+1}-w^{i}\right\|^{2}.
\end{array}
\end{equation}

We observe the relation between $\eta_{k}$ and $7M_{k}$, as well as the value of $\alpha_k$, is critical for the following analysis, and make another assumption as the following,

\begin{assumption}\label{assump3}
We assume 
\begin{equation}\label{eq8}
    \alpha_k > 0,
    \eta_{k}>7 M_{k}, k=1, \cdots, K .
\end{equation}
\end{assumption}

Next, we discuss the convergence of the augmented Lagrange function.
\begin{lemma}\label{lemma4}
    \begin{equation}
        \lim _{t \rightarrow \infty} L\left(\left\{w_{k}^{t}\right\}, w^{t}, \lambda^{t}\right) \geq -\operatorname{diam}^{2}(W) \sum_{k=1}^{K} \frac{M_{k}}{2}>-\infty,
    \end{equation}
    where the diameter of the set $W$ is defined as $\operatorname{diam}(W):=\sup \left\{\left\|w_{1}-w_{2}\right\| \mid w_{1}, w_{2} \in W\right\}$.
\end{lemma}

\begin{proof}
We express the augmented Lagrange function as follows
    \begin{equation}
        \begin{array}{l}
L\left(\left\{w_{k}^{t+1}\right\}, w^{t+1} ; \lambda^{t+1}\right) \\
=\sum_{k=1}^{K}\left(F_{k}\left(w_{k}^{t+1}\right)+\left\langle \lambda_{k}^{t+1}, w_{k}^{t+1}-w^{t+1}\right\rangle+\frac{\eta_{k}}{2}\left\|w_{k}^{t+1}-w^{t+1}\right\|^{2}\right) \\
\stackrel{(\mathrm{a})}{=} \sum_{k=1}^{K}\left(F_{k}\left(w_{k}^{t+1}\right)+\left\langle\nabla F_{k}\left(w^{[t+1](k)}\right), w^{t+1}-w_{k}^{t+1}\right\rangle\right. \\
\left.+\frac{\eta_{k}}{2}\left\|w_{k}^{t+1}-w^{t+1}\right\|^{2}\right) \\
=\sum_{k=1}^{K}\left(F_{k}\left(w_{k}^{t+1}\right)+\left\langle\nabla F_{k}\left(w^{[t+1](k)}\right)-\nabla F_{k}\left(w^{t+1}\right), w^{t+1}-w_{k}^{t+1}\right\rangle\right. \\
\left.+\left\langle\nabla F_{k}\left(w^{t+1}\right), w^{t+1}-w_{k}^{t+1}\right\rangle+\frac{\eta_{k}}{2}\left\|w_{k}^{t+1}-w^{t+1}\right\|^{2}\right) \\
\stackrel{\text { (b) }}{\geq} \sum_{k=1}^{K}\left(F_{k}\left(w^{t+1}\right)+\frac{\eta_{k}-3 M_{k}}{2}\left\|w_{k}^{t+1}-w^{t+1}\right\|^{2}\right. \\
\left.-M_{k}\left\|w^{[t+1](k)}-w^{t+1}\right\|\left\|w^{t+1}-w_{k}^{t+1}\right\|\right) \\
\geq \sum_{k=1}^{K}\left(\frac{\eta_{k}-4 M_{k}}{2}\left\|w_{k}^{t+1}-w^{t+1}\right\|^{2}-\frac{M_{k}}{2}\left\|w^{[t+1](k)}-w^{t+1}\right\|^{2}\right) \\
\stackrel{(c)}{\geq} -\operatorname{diam}^{2}(W) \sum_{k=1}^{K} \frac{M_{k}}{2} \geq-\infty . \\
\end{array}
    \end{equation}
    Expression (a) is deduced from Eq.\ref{eq12} in Lemma \ref{lemma1}. According to Cauchy-Schwartz inequality, we transfer the inequality to expression (b), And expression (c) is obtained by Assumption 2 and the definition of $\operatorname{diam}(W)$.
\end{proof}
Combine Assumption \ref{assump1}-\ref{assump3} and Lemma \ref{lemma1}-\ref{lemma4}, we conclude that the objective function $L\left(\left\{w_{k}^{t+1}\right\}, w^{t+1} ; \lambda^{t+1}\right)$ decreases at each iteration and converges to the set of stationary solutions (see Theorem \ref{theor1}).
\begin{theorem}\label{theor1}
\begin{equation}
    \begin{array}{l}
\lim _{t \rightarrow \infty}\left\|w_{k}^{t+1}-w_{k}^{t}\right\| \rightarrow 0, \forall k \\
\lim _{t \rightarrow \infty}\left\|w^{t+1}-w^{t}\right\| \rightarrow 0
\end{array}
\end{equation}
\end{theorem}
\textbf{Statement: }According to Theorem \ref{theor1}, we can conclude that the augmented Lagrangian function (Eq.\ref{alf}) can converge under the condition that Eq.\ref{eq8} is satisfied. 
Furthermore, subject to the convergence of Eq.\ref{alf}, we have derived the relationship between the training step size and latency, which is explicitly manifested in Eq.\ref{alpha} as well as Assumption \ref{assump3}.

\section{CONCLUSION}

Deep learning has been widely used in various fields, and the emerging distributed training gradually becomes popular in deep learning fields. To make distributed training resilient to updating delays, we provide a theoretical analysis on decentralized federated learning, and analyze the impact of delayed parameter updating. There are exciting avenues for further exploration and application of the derived statement to harness the power of decentralized federated learning in optimizing and enhancing performance of distributed learning.

\section*{ACKNOWLEDGMENT}
This work was supported by the National Science Foundation of China (61962045, 61502255, 61650205), the Program for Young Talents of Science and Technology in Universities of Inner Mongolia Autonomous Region (NJYT23104), the Open Foundation of State Key Laboratory of Networking and Switching Technology (Beijing University of Posts and Telecommunications) (SKLNST-2020-1-18).

\bibliographystyle{unsrt}
\bibliography{references}

\end{document}